\documentclass[12pt,fleqn,reqno]{amsart} 

\usepackage{amsmath,euscript}
\usepackage{amsthm}
\usepackage{amssymb}
\usepackage{graphicx}
\usepackage{mathrsfs}
\usepackage{epsf}
\usepackage{psfrag}
\usepackage{enumerate}
\usepackage{amsfonts,amssymb,amsthm,amsmath,latexsym,marginnote}
\usepackage{hyperref}

\usepackage[numbers]{natbib}
\usepackage{enumitem} 
\usepackage{xspace}
\usepackage[margin=1.37in]{geometry}

\newcommand{\comment}[1]{\vskip.3cm
\fbox{%
\parbox{0.93\linewidth}{\footnotesize #1}}
\vskip.3cm}

\usepackage{soul}

\newtheorem{thm}{Theorem}

\newtheorem{lemma}[thm]{Lemma}
\newtheorem{prop}[thm]{Proposition}
\newtheorem{cor}[thm]{Corollary}

\theoremstyle{remark}

\theoremstyle{definition}

\numberwithin{thm}{section}
\numberwithin{equation}{section}

\usepackage{soul}
\setstcolor{red}

\newcommand{\nc}{\newcommand}

\nc{\la}{\label}
\nc{\ba}{\begin{array}}
\nc{\ea}{\end{array}}
\nc{\bs}{\begin{split}}
\nc{\es}{\end{split}}

\newcommand{\R}{\mathbb{R}}
\newcommand{\C}{\mathbb{C}}

\newcommand{\cC}{\mathcal{C}}

\newcommand{\cF}{\mathcal{F}}

\nc{\al}{\alpha}
\nc{\del}{\delta}
\nc{\h}{\delta}
\nc{\G}{\Gamma}
\nc{\et}{\eta} 
\nc{\g}{\gamma}
\nc{\gam}{\gamma}
\nc{\ka}{\kappa}
\nc{\lam}{\lambda}
\nc{\Lam}{\Lambda}
\nc{\Om}{\Omega}
\nc{\om}{\omega}

\nc{\ta}{\tau}
\nc{\w}{\omega}
\nc{\io}{\iota}
\nc{\z}{\zeta}
\nc{\s}{\sigma}
\nc{\Si}{\Sigma}
\nc{\vphi}{\varphi}

\nc{\e}{\epsilon}

\nc{\bP}{\bar{P}}
\nc{\bQ}{\bar{Q}}

\nc{\ran}{\rangle}
\nc{\lan}{\langle}

\newcommand{\ls}{\lesssim}
\newcommand{\gs}{\gtrsim}

\newcommand{\Ran}{\operatorname{Ran}}

\newcommand{\supp}{\operatorname{supp}}

\newcommand{\re}{\operatorname{Re}}
\newcommand{\im}{{\rm Im}}

\nc{\bfone}{{\bf 1}}
\nc{\1}{{\bf 1}}

\newcommand{\p}{\partial}

\newcommand{\n}{\nabla}

\newcommand{\DETAILS}[1]{}

\newcommand{\x}{\lan x\ran}

\nc{\den}{\text{den}}
\nc{\ex}{\text{xc}}
\nc{\Ex}{\text{Xc}}

\nc{\jx}{\langle x \rangle}

\begin{document}

\title[Maximal Velocity of Propagation]{Maximal Speed of Quantum Propagation}


\author{J. Arbunich}
\address{Jack Arbunich,  Department of Mathematics, University of Toronto, 40 St. George street, Toronto, 
  M5S 2E4, Ontario, Canada}
\email{jack.arbunich@utoronto.ca}

\author{F. Pusateri}
\address{Fabio Pusateri,  Department of Mathematics, University of Toronto, 40 St. George street, Toronto, 
  M5S 2E4, Ontario, Canada}
\email{fabiop@math.toronto.edu}

\author{I. M. Sigal} 
\address{Israel Michael Sigal, Department of Mathematics, University of Toronto, 40 St. George street, Toronto,   M5S 2E4, Ontario, Canada}
\email{im.sigal@utoronto.ca}

\author{A. Soffer}
\address{Avy Soffer, Mathematics Department, Rutgers University, New Brunswick, NJ 08903, USA}
\email{soffer@math.rutgers.edu}


\date{March 3, 2021}

\begin{abstract}  
For Schr\"odinger equations with both time-independent and time-dependent Kato potentials, 
we give a simple proof of the maximal speed bound. The latter  states that the probability 
to find the quantum system outside a ball of
radius proportional to the time lapsed decays as an inverse power of time.  
We give an explicit expression for the constant of proportionality 
in terms of the maximal energy available to the initial condition. 
For the time-independent part of the interaction, 
we require neither decay at infinity nor smoothness.

\end{abstract}

\maketitle


\bigskip
\section{Introduction} 
It was shown in \cite{SigSof} that there is a constant $c$ such that  the probability to find a quantum mechanical system outside the ball of radius $c t$ 
  decays as an inverse power of time $t$. 
We call this result the {\it maximal propagation speed bound}.
It  complements the minimal propagation speed estimates obtained in \cite{SigSof, Skib, HunSigSof}; 
 see also \cite{Enss}, \cite{HuangSof} for related results, and \cite{HunSig2} for a review of the latter. 
 
In this paper, we give a simple proof of the maximal propagation speed (MPS)
bound  for Schr\"odinger (or von Neumann) equations with time-independent and time-dependent potentials.
We also give an upper bound on the infimum of such constants $c$ as above, 
which we call the {\it maximal propagation speed},  in terms of the maximal energy of the state involved. 
Our assumption on the time-independent part of the interaction is rather general
and allows e.g. many-body potentials with Coulomb singularities. 
In particular, neither decay at infinity nor smoothness is required.

 There is a certain parallel between the MPS bound and the Lieb-Robinson one, which has found many applications and has been discussed extensively in the literature, see \cite{NachSim, GebNachReSims, MatKoNaka, ElMaNayakYao} for a review 
  and some recent papers. 
  Both deal with the maximal propagation speed of information 
   in quantum systems over large distances. 
 The seeming discrepancy in that  the Lieb-Robinson  bounds do not depend on the energy is due to the fact that, with exception of \cite{GebNachReSims}, these bounds are obtained for discrete Hamiltonians for which the energy is bounded to begin with. (\cite{GebNachReSims} uses an explicit momentum cut-off in the solution.)
  
The paper is organized as follows. In Section \ref{sec:probl-res}
we formulate the problem precisely and state our first main result dealing with time-independent potentials. In Section \ref{sec:time-dep-pot} we  state the corresponding result for time-dependent potentials. In Sections \ref{sec:max-vel1-pf} and \ref{sec:max-vel-x-Ht-pf}, we prove the corresponding results and in Appendix \ref{sec:commut} we collect commutator estimates (non-Abelian functional calculus) 
based on those described in \cite{HunSig2}. 

In what follows, the relation $A\ls B$ means that there is a 
 constant $C>0$  independent of the parameters $s, t, a, b, c$ appearing below and s.t. $A\le C B$.

\section{The problem and results for time-independent potentials}\label{sec:probl-res}
We consider the Schr\"odinger equation $i \p_t \psi_t = H \psi_t$, 
with a Schr\"odinger operator $H=-\frac{1}{2}\Delta+V(x)$ on $L^2(\R^d)$ and an initial condition $\psi$. 
We will assume that $V(x)$ is $\Delta$-bounded with the relative bound $<1$, i.e. it satisfies 
 (with $\|\cdot\|$ being the norm in $L^2(\R^d)$)
\begin{align} \label{V-cond}
& \exists \,0\le a<1,\ b>0: \quad
\|Vu\|\le a \|\frac12\Delta u\|+ b\|u\|. 
\end{align}
Then by Kato's result (see e.g. \cite{CFKS}),  $H$ is self-adjoint on the domain of $\Delta$.

We say that a  state $\psi$ (or  $\psi_t=e^{-iHt}\psi$) obeys the {\it maximal propagation speed (MPS) bound} 
if there is a constant $c<\infty$ s.t.,  as $|t|\to\infty$, 
\begin{align} \label{max-speed-def}
\int_{|x|\ge c|t|}dx\,|\psi_t(x)|^2\to 0. 
\end{align}
Moreover, the scalar  $c_{\rm max}:=\inf \{ c:$ \eqref{max-speed-def} holds$\}$ will be called  the {\it maximal propagation speed}. 

Let $\chi_{A}$ stand for the characteristic function of a set $A$. In this paper, we show that any state 
 $\psi\in \Ran \chi_I(H)$ obeys the MPS bound and the maximal propagation speed $c_{\rm max}$ is bounded above by the constant 
  \begin{align} \label{k} k:=\||p|\chi_I(H)\|, 
\end{align} 
where $p:=-i\n$ and $I$ is an interval in $\R$.
 Note that, if  $I:=(-\infty, E]$ for some energy $E>0$, then due to \eqref{V-cond}, $k^2\le 2(E+b)/(1-a)$. 
 We have 

\bigskip

\begin{thm}\label{thm:max-vel-x}[Maximal propagation speed bound] 
Suppose that $H$ satisfies \eqref{V-cond}.
Let 
 $g\in C_0^\infty(I)$ be real and let  $A_\rho^\pm:=\{x\in \R^d: \pm |x|\ge \pm\rho\}$. 
 If $c> k$, then
\begin{align} \label{max-vel-est}
&\|\chi_{A_{ct+a}^+}\,e^{-iHt}g(H)\chi_{A_{b}^-}\|\ls \,t^{-n}, 
\end{align}
 for $t>1$ and any $n$, uniformly in $a$ in any region $0<b<a\ls \,t$.
\end{thm} 

\bigskip
{\it Earlier results.} 
 MPS bounds were first given in 
  \cite{SigSof} and then extended in   \cite{Skib} and \cite{BonyFaupSig} 
and used in scattering theory in \cite{SigSof2, Sig, HeSk, DerGer1, FaupSig}. 
 
 The main   contributions of this paper compared to earlier work are:

(a) 
 a significant extension of the class of potentials (in particular, no smoothness, or decay at infinity is required); 

(b) a much simpler proof; 

(c) the determination of the precise maximal velocity bound (see \eqref{k}).
 
\bigskip
{\it Light cone.} 
Inequality \eqref{max-vel-est} (see also \eqref{max-vel-est-weight}) implies that the probability to find the quantum system outside the ball $\x\le ct+a$ decays with time. Moreover, for well localized initial conditions, it gives a rate of decay.

Below, we prove a stronger result implying a limitation on the speed of propagation of information.  
Let $\al_t$ be the Heisenberg evolution, $\al_t(B):= e^{iHt}B e^{-iHt}$. Then we show for $c> k, \rho$ satisfying $\rho> a+ct$, for $a$  sufficiently large,
\begin{align} \label{max-vel-est-info}
&\|\chi_{A_{\rho}^+}\,\al_t(\chi_{A_{b}^-}g(H))\|\ls \,\rho^{-n}, 
\end{align}
i.e. it takes at least $\rho/c$ units of time for a signal originating in $A_{b}^-$ to reach $A_{\rho}^+$. 

\bigskip
{\it Density matrices.} 
Since Theorem \ref{thm:max-vel-x} and its version in \eqref{max-vel-est-info} 
deal with propagation of observables, they are valid also for states described by density operators, 
i.e. for the von Neumann dynamics.

\bigskip
{\it Weighted estimates.} 
Inequality 
\eqref{max-vel-est} implies weighted estimates.  
Indeed, using 
\[\lan x\ran^{-\al} = \lan x\ran^{-\al}\chi_{A_{\e t/2}^-} +O(t^{-\al}),\] 
and using \eqref{max-vel-est}, with $b=\e t/2$, we find, for $n\ge \al$, 

 
 
\begin{align} \label{max-vel-est-weight}
\|\chi_{A_{(c+\e)t}^+}\,e^{-iHt}g(H)\lan x\ran^{-\al}\|\ls \,t^{-\al}. \end{align}

\bigskip
{\it Abstract formulation.} Theorem \ref{thm:max-vel-x} could be stated in an abstract
setting for a pair of self-adjoint operators $H$ and $X$ with multiple commutators required to be  $H$-bounded, with some extra care in the definition of
the commutators for unbounded operators. Here we deal with
Schr\"odinger operators $H=\frac{1}{2}p^2+V(x)$ and with the
particle position coordinate $x$, more precisely, with $\lan x\ran$.

\bigskip
{\it Time reversal.} The results for $t<0$ are 
obtained  by complex conjugation 
$\cC:\,\psi\to\overline\psi$, using
$\cC H\cC=H,\ \cC x\cC=x$: 
  \eqref{max-vel-est} transform into
 \begin{align}  \label{time-rev2} 
\|\chi_{A_{ct+a}^+}\,e^{iHt}g(H)\chi_{A_{b}^-}\|\ls \,t^{-n}\end{align}
for $t>1$ and for any $n$, which can be reformulated in terms of  $t<-1$.

 \bigskip
 
{\it Approach.}   Our approach can be thought of as microlocal analysis 
 with an operator functional calculus replacing the  pseudodifferential one. 
It is based on the method of propagation observables which we now explain.

We consider a time-dependent, non-negative 
operator-family ({\it propagation observable}) $\Phi_t$. We would like to obtain {\it propagation estimates} 
of the form $\|\Phi_t\psi_t\|\ls t^{-m}$. 
Denote the inner product in $L^2(\R^d)$ by $\lan \cdot, \cdot\ran$, so that $\|\cdot\|=\sqrt{\lan \cdot, \cdot\ran}$ and let $\psi_t=e^{-iHt}g(H)\phi$ be a spectrally localized solution to 
the Schr\"odinger equation and let $\lan A\ran_t :=\lan\psi_t, A\psi_t\ran$. Note the relation  
\begin{align}
\label{dt-Heis}
&{d\over{dt}}\left<\Phi_t\right>_t =\lan D\Phi_t\ran_t; \qquad D\Phi_t=i[H,\Phi_t]+{\partial\over{\partial t}}\Phi_t.
\end{align}
We call $D$ the {\it Heisenberg derivative}.  
Using $\lan \Phi_t\ran_t= \lan \Phi_0\ran_0+\int_0^t \p_r\left<\Phi_r\right>_r dr$, we find 
\begin{align} \label{eq-basic}  
\lan \Phi_t\ran_t-\int_0^t \lan D\Phi_r\ran_r dr= \lan \Phi_0\ran_0,
\end{align}
which we call the {\it basic equality}.
If $g(H) D\Phi_t g(H) \le 0$, modulo fast time-decaying terms, 
then, after pulling $g(H)$ out of $\psi_t$, 
this relation gives estimates on the positive terms $ \lan \Phi_t\ran_t$ and $-\int_0^t \lan D\Phi_r\ran_r dr$.

We will consider propagation observables of the form 
$\Phi_{ts}= f(x_{ts})$,  where $x_{ts}:=s^{-1}(\x -a-c t)$, with $s\ge t$, and 
 $f$ is a non-negative, non-decreasing function supported in $(0, \infty)$.
The factor $s^{-1}$ is introduced to control multiple commutators and commutator products. 
It can be thought of as an adiabatic or semi-classical parameter. 

The energy cut-off $g(H)$ appearing next to $D\Phi_t$ is  pulled out of $\psi_0$, 
by commuting it with the evolution $e^{-iHt}$. 
The latter simple but crucial fact expresses the conservation of energy for the Schr\"odinger equation. 
It is not valid for Schr\"odinger equations with time-dependent Schr\"odinger operators, $H(t)$. 
This is the main obstacle in treating  such equations. 
We overcome it by introducing the {\it asymptotic energy cut-offs} 
and proving the {\it pull-through relation}, replacing the commutativity of $g(H)$ and $e^{-iHt}$. 
In what follows, we let $g\in C_0^\infty(I)$.

\bigskip
\section{Time-dependent potentials}\label{sec:time-dep-pot}
Consider time-dependent hamiltonians of the form 
\begin{align}\label{Ht}
H_t:=H+W_t,
\end{align}
where $H$ is as above and $W=W(x, t)$ is a real, time-dependent potential satisfying 
\begin{align}\label{Wt-cond}
\p^\al_x W_t(x)=O(\langle t \rangle^{-\mu-|\al |}),\ \quad \mu>1,\ |\al|\le 2.
\end{align}
Here $\langle t \rangle:=(1+t^2)^{1/2}$, but in what follows we consider only $t>0$ 
and ignoring the singularity at $t=0$, we write $t$ for $\langle t \rangle$.

Let $U_t := U(t, 0)$ be the evolution generated by $H_t$. 
Here the crucial role is played by the asymptotic energy cut-offs introduced in \cite{SigSof}:

\begin{prop}\label{prop:as-en-cutoff}[Asymptotic energy cut-off]   Let condition \eqref{Wt-cond} hold but with $\mu>0$ instead of $\mu>1$.
Then the following operator-norm limit exists
\begin{align} \label{g+}
g_+(H) := \lim_{t\rightarrow \infty}  U_t^{-1}g(H)U_t.
\end{align}
 \end{prop}
 
\begin{proof}  
Define $g_t(H):=U_t^{-1}g(H)U_t$, 
and write $g_t(H)$ as the integral of the derivative and use that $\p_r g_r(H) = iU_r^{-1}[g(H), W_r]U_r$ to obtain
\begin{align} \label{gt}
g_t(H)=g(H)+i\int_0^t U_r^{-1}[g(H), W_r]U_rdr.
\end{align}
Since $[g(H), W_r] =O(r^{-\mu-1})$ (see Lemma \ref{lem:g-W-comm} of Appendix \ref{sec:commut}), 
this shows that $g_+(H)=\lim g_t(H)$ (the operator norm limit) exists provided $\mu>0$.  \end{proof}

{\bf Remark.} As was noticed by one of the referees, (a) assumption \eqref{Wt-cond} 
 on the time-dependent potential $W_t$ can be relaxed, e.g. to 
 $\int_0^\infty  \|\p_x^\al W_t(x)\| dt <\infty, |\al|\le 2$; (b) for $\mu > 1$ 
 (which is used in our proofs below)
 the proof of Proposition \ref{prop:as-en-cutoff} is simpler and the obvious estimate 
 $\|[g(H), W_r]\|\le 2\|g(H)\| \|W_r\| =O(r^{-\mu})$ 
 is sufficient. So, we could have assumed e.g. $\|W_t (H+1)^{-1}\|_{L^\infty}\ls  \lan t\ran^{-\mu},\ \mu>1$. 

Eq. \eqref{gt} implies that 
\begin{align}\label{g+gt} 
g_+(H)-g_t(H)= i\int_t^\infty U_r^{-1}[g(H), W_r]U_rdr=O(t^{-\mu}),
\end{align}
which, together with $U_t g_t(H)=U_tU_t^{-1}g(H)U_t=g(H)U_t$,  implies

\begin{cor}\label{prop:pull-thru}[Pull-through relation]
With the notation and assumptions above, see \eqref{Ht}, \eqref{Wt-cond} and \eqref{g+}, we have
\begin{align} \label{PTR} 
U_t g_+(H)= g(H)U_t+O(t^{-\mu}).
\end{align}
\end{cor}

\smallskip
Using the above key relation with $\mu>1$,
and additional commutator estimates, we are able to 
extend Theorem \ref{thm:max-vel-x} to time-dependent potentials:  

\smallskip
\begin{thm}\label{thm:max-vel-x-Ht}[Maximal propagation speed for $t$-dependent potentials] 
Suppose that $H_t$ is of the form \eqref{Ht}, with $H$ and $W_t$ satisfying \eqref{V-cond} and \eqref{Wt-cond}.
Let 
 $g\in C_0^\infty(I)$ be real. 
 If $c> k$, then  
\begin{align} \label{max-vel-est-Ht}
\|\chi_{A_{ct+a}^+}\,U_t g_+(H)\chi_{A_{b}^-}\|\ls \,t^{-1/2} 
\end{align}
for $t>1$, uniformly in $a$ in any region $0<b<a\ls \,t$. Here $A_{\rho}^\pm$ are the sets defined in Theorem \ref{thm:max-vel-x}.\end{thm} 

The next result from \cite{SigSof} shows that the initial conditions of the 
form $g_+(H)\chi_{A_{b}^-}\phi$, 
for various $g (H), b>0$ and $\phi\in L^2$,
form a dense set. 

\begin{thm}\label{thm:g+-approx} 
Let $(I_n) \subset \R$ be a sequence of increasing intervals with $\lim_{n \rightarrow \infty}I_n = \R$,
and let $g_n \in C_0^\infty(\R)$ be such that $g_n \equiv 1$ on $I_n$. 
Let $g_{n,+} := \lim U_t^{-1}g_n(H)U_t$, as in \eqref{g+}.
Then
\begin{align*}
s-\lim_{n \rightarrow \infty} g_{n,+}(H) = \mathrm{id}.
\end{align*}
\end{thm}

See \cite{SigSof} for  the proof of this theorem. 

\medskip
As was mentioned after \eqref{eq-basic}, to bound  $D\Phi_{t s}$ from above, 
we have  `to pull  $g(H)$ from $\psi_0$' through $U_t$. 
For the time-dependent situation this is achieved using \eqref{PTR}.
Note that we need  $\mu>1$ for the remainder term in \eqref{PTR} to lead to an integrable contribution.

\bigskip
\section{Proof of Theorem \ref{thm:max-vel-x}}\label{sec:max-vel1-pf} 

\DETAILS{We prove that, for any smooth, {\bf (FP) non-decreasing; IM: we do not need this}, bounded function $\chi$ 
supported in $\R^+$ and for any $ v> k, a>b>0$ and $s\ge t$, 
\begin{align} \label{max-vel-est2}
\|\chi(s^{-1}(\x-a -v t))\,e^{-iHt}g(H)\chi_{A_{b}^-}\| 
\ls \,s^{-n},
\end{align} 
where $n$ is any positive number. 

{\bf (FP) have created a little notation section below.}

\comment{
I think the $\cF_\delta$ notation and the admissibility definition should be clarified, and could be simplified.

For example, take a fixed $f$ smooth, $f'\geq 0$ etc... with $f(\lambda)=1$ for $\lambda \geq 1$.
Then look at $f_\del(\lambda) := f(\del^{-1}\lambda)$, for $\del := v-k$ and work with this.

Note that in the argument the constants degenerate as $\delta\rightarrow 0$ but we don't track it, or say it.
Then allowing this dependence, we can say that $\tilde{f}$ admissible means 
$\tilde{f} \geq 0$ (smooth) and that the support is contained in the support of $f_\delta'$
(which is the same as $\tilde{f} \ls f_\delta'$.}}

\smallskip
{\it Definitions and notation}. 
We fix $c> v$ and 
 let $\cF$ be the set of functions  $0\le f\in C^\infty(\R)$, supported in $\R^+$ and
satisfying $f(\lam)=1$ for $\lam\ge c-v$,  and $f^\prime\ge 0$, with $\sqrt{f'}\in C^n$. Here and in the rest of the section, $n$ is the same as in \eqref{max-vel-est}.

 We say a function $h$ is {\it admissible} if it is smooth, non-negative with $\supp h\subset (0, c-v)$ and $\sqrt{h}\in C^n$. 
Note that if $h$ is admissible, then $h= f'$, with 
\[f(\lam)=\int_{-\infty}^\lam h(s)ds \ \text{ and } \ \Big( \int_{-\infty}^\infty h(s)ds\Big)^{-1} \, f \in \cF.\]   
 
Moreover, we will use the following notation (for any $\phi\in L^2, \|\phi\|=1$): 
\begin{align*}
\chi^-_b:=\chi_{A_{b}^-},\quad \psi_t :=e^{-iHt}g(H)\chi^-_b\phi \quad \text{and}
 \quad x_{ts} :=s^{-1}(\x -a -v t)  
\end{align*}
and the convention that 
$A\dot\le B$ and $A\dot\ls B$ mean that for any $n>0$, there is $C>0$ s.t. 
$A \le B+ C s^{-n}$ and $A \ls B+ C s^{-n}$, respectively.

We begin the proof of Theorem \ref{thm:max-vel-x} with the velocity bound below.

\begin{lemma}\label{lem:p-est}   
Given $f \in \cF$, let $u^2=f'$.
Then, for any $n$, there is $\tilde u$, with $\tilde u^2$ admissible, s.t.
\begin{align}\label{p-est}
\|p u(x_{ts})g(H)\psi\|&\,  \dot\le \, k\| u(x_{ts})g(H)\psi\|+ s^{-1}\| \tilde u(x_{ts})g(H)\psi\|.
\end{align}
\end{lemma}

\begin{proof}  
Writing $g(H)=\tilde g(H)g(H)$, with $\tilde g\in C^\infty_0(I)$, with $\tilde g=1$ on supp$g$, 
and commuting $\tilde g(H)$ to the left, we find
\begin{align}\label{p-rel} 
p u(x_{ts})g(H)&=  p\tilde g(H) u(x_{ts})g(H)+p [u(x_{ts}), \tilde g(H)] g(H). 
\end{align}
Now, by \eqref{comm-exp-x},  
$p[u(x_{ts}), \tilde g(H)]=\sum_{k=1}^{n-1}{s^{- k}\over{k!}}p B_k u^{(k)}(x_{ts}) +O(s^{-n})$ 
for any $n$, with $p B_k$ bounded. With this, taking the norm of \eqref{p-rel}  and using $\| p\tilde g(H)\|\le k$, 
we arrive at 
\eqref{p-est} with $ \tilde u^2:= C(u^2+\sum_{k=1}^{n-1}s^{- k}(u^{(k)})^2)$, for some $C>0$.
\end{proof} 

Next is a key statement in the proof of Theorem \ref{thm:max-vel-x}:

\begin{prop}\label{prop:propag-est1} 
Under the hypothesis of Theorem \ref{thm:max-vel-x}, for any $f\in \cF$ and any $n$, there is $\tilde f\in \cF$ s.t., for $s\ge t$, 
 \begin{align}\label{propag-est1} 
\int_0^t \lan f'(x_{rs})\ran_r \,dr \ls s \lan \tilde f'(x_{0s})\ran_0+O(s^{-n}). 
\end{align} 
\end{prop}

\begin{proof}
We proceed inductively, showing that \eqref{propag-est1} holds for  $n=0$ and then  assuming it  
holds for some $n=n'\in \R$ and proving that it holds for $n=n'+1$, 
and proceed in this way until we reach an arbitrary $n$. 
To this end, we use  the time dependent observable 
\begin{align}\label{propag-obs1}
\Phi_s(t) & = f(x_{ts}), \quad f \in \cF, 
\end{align}
with $0\le t\le s$.  In order to estimate $\left<\Phi_s(t)\right>_t=\lan\psi_t,\,\Phi_s(t)\psi_t\ran$,  
we apply \eqref{dt-Heis} and  the basic equality \eqref{eq-basic}. 
We start by computing $D\Phi_s(t)$. First, we have
\begin{align} \label{dt-Phi}
{\partial\over{\partial t}}\Phi_s(t)=-s^{-1}v \,f^\prime(x_{ts}).
\end{align}
Then, we let $\g : = {1\over 2}(p\cdot \n\x+\n\x\cdot p)$, with $p:=-i\n$.
Factorizing $f^\prime= u^2$ and using that $[H, \x]=\g$ and
$[[\g,u],u]=0$, we find
\begin{align} \label{H-Phi-comm}
i[H,\Phi_s(t)]=&{i\over 2}[p^2,\Phi_s(t)]
=\frac12 s^{-1}(\g f^\prime(x_{ts})+f^\prime(x_{ts})\g)
\notag
\\ 
& = s^{-1}\,u(x_{ts})\,\g\,u(x_{ts}).
\end{align}

This equation, together with Eq. \eqref{dt-Phi}, yields:
\begin{align} \label{DPhi-expr-x}
D\Phi_s(t)= s^{-1}\,u(x_{ts})\,(\g-v)\,u(x_{ts}).\end{align}
Now, we claim that, with $k$ defined in \eqref{k}, there is $C>0$ s.t.  
\begin{align}\label{A-est}\notag 
g(H)u(x_{ts})\,\g\,u(x_{ts})g(H)\le k & g(H)u(x_{ts})^2g(H)
\\
& + C  s^{-1}g(H)\tilde u(x_{ts})^2g(H),
\end{align} 
where $\tilde u(x_{ts})^2$  is an admissible function. 
To see this, we first estimate 
\begin{align} 
\label{A-est2}|\lan \psi, g(H)u(x_{ts})&\,\g\,u(x_{ts})g(H) \psi\ran|
\notag
\\ 
& \le \|\n\x u(x_{ts})g(H)\psi\|\|p u(x_{ts})g(H)\psi\|.
\end{align} 
This inequality, together with \eqref{p-est},  gives 
\begin{align}\label{p-est1}
|\lan \psi, \, & g(H)u(x_{ts})\,\g\, u(x_{ts})g(H) \psi\ran|\notag
\\ 
& \, \dot\le \,  \| u(x_{ts})g(H)\psi\| \Big(k\| u(x_{ts})g(H)\psi\|+s^{-1}\| \tilde u(x_{ts})g(H)\psi\|\Big),
\end{align} 
which implies \eqref{A-est}.
Now, using \eqref{A-est}, together with \eqref{DPhi-expr-x} 
and the definitions $u(x_{ts})^2=f'(x_{ts})$ and $h(x_{ts}):=\tilde u(x_{ts})^2$, 
we obtain
\begin{align}\label{DPhi-est1}
g(H)&D\Phi_s(t)g(H)\notag
\\ 
&  \dot\le \, (k -v)s^{-1} g(H)f'(x_{ts}) g(H) + C s^{-2} g(H)h(x_{ts}) g(H).
\end{align}

Taking the matrix element of this inequality with the vector 
$\tilde\psi_t=e^{-iHt}\chi^-_b\phi$ and using 
 \begin{align} \label{pt-triv}g(H)\tilde\psi_t=e^{-iHt}g(H)\chi^-_b\phi=:\psi_t,\end{align} we find that   
\[\lan D\Phi_s(t)\ran_t\, \dot\le \, (k -v)s^{-1} \lan f'(x_{ts}) \ran_t + C s^{-2} \lan h(x_{ts}) \ran_t.\]
Hence, Eqs. \eqref{dt-Heis} and \eqref{eq-basic} and the definition $\Phi_{ts}:= f(x_{ts})$ give  
\begin{align} \label{propag-est2} 
\lan f(x_{ts})\ran_t+(v-k)s^{-1}&\int_0^t \lan  f'(x_{rs})\ran_r dr\notag\\
& \dot\le \, \lan f(x_{0 s})\ran_0+C s^{-2}\int_0^t \lan  h(x_{rs})\ran_r dr. 
\end{align}
Since $k < v$ and since $h=\tilde u^2$  is an admissible function and therefore $h\gs j'$ for some $j\in \cF$, \eqref{propag-est2} implies that
\begin{align}\label{propag-est3} 
\int_0^t \lan  f'(x_{rs})\ran_r dr\, 
  \dot\ls \, 
  s \lan f(x_{0 s})\ran_0
  + s^{-1}\int_0^t \lan  j'(x_{rs})\ran_r dr. 
\end{align}
(The constant entering the relation $\ls$ here is bounded by a power of $(v-k)^{-1}$.)
This, together with the boundedness of $j'$,  
gives estimate \eqref{propag-est1} with $n=0$.
Now, assuming \eqref{propag-est1} holds for some $n=n'\ge 0$, 
  and using this (with $f=j$) for the integral on the r.h.s. of \eqref{propag-est3}, we see that  
\eqref{propag-est1} holds for $n=n'+1$.  
\end{proof}

\begin{proof}[End of the proof of Theorem \ref{thm:max-vel-x}]
 Recall that $\supp f\subset \R^+$ for any $f\in \cF$, 
and therefore $\supp f(x_{0 s})\subset \{\x>a+\delta s\}$. 
Since $\supp \chi^-_b\subset \{\x\le b\}$ and $b<a$, the functions $f(x_{0 s})$ and $\chi^-_b$ have disjoint supports. 
Hence, we have by \eqref{comm-exp-x} of Lemma \ref{lem:commut-exp} 
(and the fact that $f^{(k)}(x_{0s})\chi^-_b=0$) 
that, for any $n$, 
\begin{align}\label{local-est3}
f(x_{0 s})g(H)\chi^-_b=[f(x_{0 s}),g(H)]\chi^-_b=O(s^{-n})
\end{align}
and therefore, for any $n$ and any $f\in \cF $, 
\begin{align}\label{local-est4}\lan f(x_{0 s})\ran_0=\lan g(H)\chi^-_b\phi,
f(x_{0 s})g(H)\chi^-_b\phi\ran=O(s^{-n}).\end{align} 

Next, retaining the first term in \eqref{propag-est2} and dropping the second one and using \eqref{propag-est1} and \eqref{local-est4}, we conclude that, for any $n$, 
\begin{align} \label{propag-est4} 
\lan f(x_{ts})\ran_t \dot\ls \, s^{-n}.  \end{align}
Now, for any $f\in \cF $, we have $f(\lam)=1$ for $\lam\ge c-v$, and therefore $f(x_{t s})=1$ on $\{\x\ge a+vt +(c-v) s\}$.
Recalling that $s\ge t$ 
 and $A_\rho^\pm:=\{x\in \R^d: \pm \x \ge \pm\rho\}$, we have  
\[A_{ct+a}^+ \subset \{f(x_{t t})=1\}.\] 
Hence, we conclude that $\lan \chi_{A_{ct+a}^+}\ran_t \dot\ls \, t^{-n}$. Since $ \psi_t :=e^{-iHt}g(H)\chi^-_b\phi$ for any $\phi\in L^2, \|\phi\|=1$, this implies \eqref{max-vel-est}. \end{proof}

\begin{proof}[Proof of Eq. \eqref{max-vel-est-info}]
By our assumption, $\rho> a+ct$. 
 We  set $s=(\rho-a)/c\ge t$. This gives $\rho=a +c s\ge a +c t$ and therefore   
$A_{\rho}^+\subset  \{\x\ge a+vt +(c-v) s\}\subset\{f(x_{t s})=1\}$.   Using this together with estimate \eqref{propag-est4} in the time-reversed form  $\|f(x_{ts})^{1/2}\,e^{iHt}g(H)\chi_{A_{b}^-}\|\ls \,s^{-n}$,  and the definitions  $s=(\rho-a)/c$  and  $\al_t(B):= e^{iHt}B e^{-iHt}$, we arrive at \eqref{max-vel-est-info}. \end{proof}

\bigskip
\section{Proof of Theorem \ref{thm:max-vel-x-Ht}}\label{sec:max-vel-x-Ht-pf}

We follow the proof of Theorem \ref{thm:max-vel-x}. We use the notation $A\dot\le B$ and $A\dot\ls B$ introduced in the previous section. 
Let $$\psi_t := U_tg_+(H)\phi_0, \qquad \phi_0 := \chi^-_b \phi.$$
We begin by using  \eqref{dt-Heis} as before and observe that the identities
\eqref{DPhi-expr-x} and \eqref{DPhi-est1} still hold for $H_t$. However, the relation \eqref{pt-triv}  (the trivial pull-through formula) fails. 
Now, instead of \eqref{pt-triv}, we use the pull-through relation \eqref{PTR} 
to pull in $g(H)$ from $\psi_t$. 
For convenience denote
$R := \psi_t - g(H)U_t\phi_0 = O(t^{-\mu})\phi_0$.
Then
\begin{align}
\label{Ht-pf1}
\begin{split}
& \p_t\lan \Phi_{ts}\ran_t 
  = \lan g(H)U_t\phi_0, D\Phi_s(t) g(H)U_t\phi_0\,\ran 
  \\
  & + \lan R, D\Phi_s(t) g(H) U_t \phi_0\,\ran 
  +  \lan g(H)U_t\phi_0, D\Phi_s(t) R\,\ran
  +  \lan R, D\Phi_s(t) R\,\ran.
\end{split}
\end{align}

We claim that the second line above is $O(s^{-1} t^{-\mu})$.
This follows from \eqref{DPhi-expr-x} and \eqref{p-est}  
 for the first two terms; 
for the third one we can get the better bound $O(s^{-1}t^{-2\mu})$ using
\begin{align*}
\| p (U_tg_+(H) - g(H) U_t) \| = O(t^{-\mu}),
\end{align*}
which follows from \eqref{g+gt} and  estimates \eqref{p-rep} and \eqref{p-U-comm-est} of Appendix \ref{sec:commut}. 

Going back to \eqref{Ht-pf1}, and using also \eqref{DPhi-est1} to bound the first terms on the r.h.s., we get, for some admissible function $\tilde f$,
\begin{align}
\notag
\p_t\lan \Phi_{ts}\ran_t 
& = \lan g(H)U_t\phi_0, \, D\Phi_s(t)g(H)U_t\phi_0\,\ran 
  + O\big(s^{-1} t^{-\mu}) \\
\notag
&\dot\le \, (k-v)s^{-1} \lan g(H)U_t\phi_0, \,f'(x_{ts})\,g(H)U_t\phi_0\ran 
  \\
\label{DPhi-est1-Ht}
  & + Cs^{-2} \lan g(H)U_t\phi_0, \,\tilde f(x_{ts})\,g(H)U_t\phi_0\ran
  + O(s^{-1}  t^{-\mu}). \end{align}

Next, 
passing back to $\psi_t$ by using the pull-through relations in the opposite direction 
and the fact that $\tilde{f}$ is bounded, we obtain
\begin{align} \label{DPhi-est2-Ht}
\p_t\lan \Phi_{ts}\ran_t \, \le \,& (k-v)s^{-1} \,\lan f'(x_{ts})\ran_t + Cs^{-2}   + C s^{-1} t^{-\mu}.
 \end{align}
Since $v>k$, we can drop the first term on the r.h.s. and use the basic equality \eqref{eq-basic}, 
the definition $\Phi_{ts}:= f(x_{ts})$ and the conditions $\mu>1$, $s\geq t$, to find 
\begin{align} \label{propag-est2-Ht} 
\lan f(x_{ts})\ran_t\le &\lan f(x_{0s})\ran_0
+ C s^{-1}. 
\end{align}
For the first term on the r.h.s., we claim 
\begin{align} \label{local-est3-Ht} 
\lan f(x_{0s})\ran_0=O(s^{-2\mu/(\mu+1)}).\end{align} 
To prove this estimate,  
we recall  $\psi_t:=U_t g_+(H)\chi^-_b\phi$, note that $\lan f(x_{0s})\ran_0=\| \chi(x_{0s})\psi_0\|^2$, 
with $\chi^2= f$, and pass from $g_+(H)$ to $g_{s^\beta}(H):=U_{s^\beta}^{-1} g(H) U_{s^\beta}$, 
with $\beta<1$, paying with the error $O(s^{-\beta \mu})$:
\[\chi(x_{0s})\psi_0=\chi(x_{0s})g_+(H)\chi^-_b\phi=\chi(x_{0s})g_{s^\beta}(H)\chi^-_b\phi+O(s^{-\beta \mu}).\] 
In Lemma \ref{lem:chi-g+-est} of Appendix \ref{sec:commut} we show that \begin{align}
\label{chi-g+-est} 
\chi(x_{0s})g_{s^\beta}(H)\chi^-_b =O(s^{ \beta-1}).
\end{align} 
This, together with the previous estimate 
and the choice $\beta=1/(\mu+1)$, yields (after squaring up) \eqref{local-est3-Ht}.
\eqref{local-est3-Ht} and \eqref{propag-est2-Ht} imply
\begin{align*}
\lan f(x_{ts})\ran_t\le Cs^{-1},
\end{align*} 
which, in view of the definition of $f$, gives, after setting $s=t$, Theorem \ref{thm:max-vel-x-Ht}.
$\Box$ 

\bigskip {\bf Acknowledgement.} We are grateful to anonymous referees for useful remarks and suggestions. The work on this paper was supported in part by NSERC Grant No. NA7901 (JA and IMS), by a start-up grant from the University of Toronto and NSERC Grant No. 06487 (JA and FP) and by NSF grants DMS-1600749 and NSFC11671163 (AS). 
None of the authors has a conflict of interest.  
  
\appendix 

\section{Commutator expansions}\label{sec:commut}
In this appendix, we present commutator expansions and estimates, first derived in \cite{SigSof} 
and then improved in \cite{Skib,HunSig1,HunSigSof} (see \cite{Dav,HelffSj,IvrSig} for the original work).
We follow \cite{HunSig1} and refer to this paper for 
details and references.
Here, we mention only that, by the Helffer-Sj\"ostrand formula, 
a function $f$ of a self-adjoint operator $A$ can be written as
\begin{align} \label{fA-repr}
&f(A)=\int d\widetilde f(z)(z-A)^{-1},
\end{align}
where $\widetilde f(z)$ is an almost analytic extension of $f$ to $\C$ supported in a complex neighbourhood of $\supp f$.
For $f\in C^{n+2}(\R)$, we can choose $\widetilde f$ satisfying the estimates (see (B.8) of \cite{HunSig1}):
\begin{align} \label{tildef-est}
&\int |d\widetilde f(z)||\im(z)|^{-p-1}\ls \sum_{k=0}^{n+2}\|f^{(k)}\|_{k-p-1},\end{align}
where $\|f \|_{m}:=\int \x^m |f(x)|$.
Note that \cite{HunSig1} requires $f\in C^\infty_0(\R)$, 
but one can easily extend representation \eqref{fA-repr} and the needed results to $f$'s satisfying 
\[\sum_{k=0}^{n+2}\|f^{(k)}\|_{k-2}<\infty,\] for some $n\ge 1$, which covers the case $f\in C^\infty(\R)$ with $f'\in C^\infty_0(\R)$ which comes up in this work.

The essential commutator estimates are incorporated in the following lemma:

\begin{lemma}\label{lem:commut-exp} %
Let $f\in C^\infty(\R)$ be bounded, 
  with $\sum_{k=0}^{n+2}\|f^{(k)}\|_{k-2}<\infty$, for some $n\ge 1$.
Let  $x_s=s^{-1}(\x-a)$ for $a >0 $ and
$1\le s<\infty$. Suppose that $H$ satisfies \eqref{V-cond} and let 
 $g\in C_0^\infty(\R)$. Then, for any $n\ge 1$, 
 \begin{align}\label{comm-exp-x} [g(H), f(x_s)]&=  \sum_{k=1}^{n-1}{s^{- k}\over{k!}}B_k f^{(k)}(x_s) +O(s^{-n}), \end{align}
 uniformly in $a\in \R$, where $H^j B_k, j=0, 1, k =1, \dots, n-1,$ 
 are bounded operators and $\|H^j O(s^{-n})\|\ls s^{-n}, j=0, 1$. 
 For $n=1$, the sum on the r.h.s. is omitted. 
\end{lemma}

\begin{proof}
Omitting the argument $x_s$ in 
$f$ and $f^{(k)}$ and using (B.14)-(B.15) of \cite{HunSig1}, we have
 \begin{align}
\label{comm-exp1}[g(H), f]&=\sum_{k=1}^{n-1}{s^{- k}\over{k!}}B_k f^{(k)}+s^{-n}Re(s);\\
\label{Bk-bnd}B_k&= \,{ad_{\x}^kg(H)};\\ 
\label{Res-exp} Re(s)&=\int d\widetilde f(z)(z-x_s)^{-1}B_n(z-x_s)^{-n}, 
     \end{align}    
where for $n=1$, the sum on the r.h.s. is omitted. 
Now we show that the operators $H^j B_k, j=0, 1, k =1, \dots, n,$ and  $Re(s)$ satisfy the estimates:  
\begin{align} \label{Rem-est}
\|H^j B_k\|\ls 1,\ j=0, 1, k =1, \dots, n, \quad \|H^j Re(s)\|\ls 1, 
\end{align} 
uniformly in $a$ (and of course $s$).  
Let $R(z):=(z-H)^{-1}$. We claim that
\begin{align} \label{Bk}
& B_k=\sum_{\pi\in \Pi_k}\int d\widetilde g(z)R(z)C_{k_1} \dots R(z)C_{k_s} R(z),
\end{align}
where $\pi=(k_1, \dots, k_s)$, with $1\le k_j\le 2$ and $k_1 + \dots + k_s=k$, 
are {\it ordered} partitions of $k$,  
and $C_{j}:=ad_{\x}^j H$. 
Indeed, using the formula \eqref{fA-repr} for $g(H)$, \begin{align} \label{gH-repr}
&g(H)=\int d\widetilde g(z)R(z),\end{align} 
and the relation  $ad_{\x}R = RC_1R$, we see that this is true for $k=1$:
 \begin{align} \label{B1}
&B_1=\int d\widetilde g(z)R(z)C_1 R(z).\end{align} 
  Now, assuming that \eqref{Bk} holds for $k=m$, we prove it for $k=m+1$. To this end, we use   that $ad_{\x}B_j=B_{j+1}$, $ad_{\x}R = RC_1R$ and $ad_{\x}C_j=C_{j+1}$ to obtain 
\begin{align} \label{Bm}
B_{m+1}=&\sum_{\pi\in \Pi_m}\int d\widetilde g(z)\bigg\{\sum_{a=1}^m \prod_{j=1}^{a-1}(R(z)C_{k_j})(R(z)C_1R(z)C_{k_a}+R(z)C_{k_a+1})\notag\\
&\qquad \times \prod_{j=a+1}^{s}(R(z)C_{k_j}) R(z)+\prod_{j=1}^{s}(R(z)C_{k_j})(R(z)C_1R(z))\bigg\}\notag\\
&=\sum_{\pi\in \Pi_m}\int d\widetilde g(z)\bigg\{\sum_{a=1}^m \big[\prod_{j=1}^{a-1}(R(z)C_{k_j})(R(z)C_1)\prod_{j=a}^{s}(R(z)C_{k_j})\notag\\
&\qquad +\prod_{j=1}^{a-1}(R(z)C_{k_j})(R(z)C_{k_a+1}) \prod_{j=a+1}^{s}(R(z)C_{k_j})\big]\notag\\
&\qquad +\prod_{j=1}^{s}(R(z)C_{k_j})(R(z)C_1)\bigg\}R(z).\end{align}
It is not hard to see that  \eqref{Bm} is of the form  \eqref{Bk} with $k=m+1$.
Finally, since  $ad_{\x}^k H=0$, for $k\ge 3$, we have that   
either $k_a+1\le 2$ or $C_{k_a+1}=0$.

To prove boundedness of $B_k$, we use that $C_1:=[\x, H]=i\g$, where, 
recall, $\g={1\over 2}(p\cdot \n\x+\n\x\cdot p)$, with $p:=-i\n$, and $C_2:=[\x, [\x, H]]=- |\n\x|^2$. 
Let \begin{align} \label{S-def}S:=(H+c)^{1/2},\ \text{ where }\ c:=-\inf H+1,\end{align} and $(H +c)^{s}, s\in \R,$ is defined  by the spectral theory. Since $\g$ is $S$-bounded, we have $\|\g R(z)\|\ls \|SR(z)\|\ls \sup_{s\ge 1}(s^{1/2}|s-c-z|^{-1})$, which gives
\begin{align} \label{C1-est}\|C_1R(z)\|\ls \lan\re z\ran^{1/2}/|\im z|.\end{align} Using this and $|\n\x|\le 1$ in  \eqref{Bk} and using \eqref{tildef-est} for $g$ and the fact that $\widetilde g(z)$ supported in a complex neighbourhood of $\supp g$, which is a compact set, shows that $H^j B_k, k =1, 2, \dots$, $j=0,1$, are bounded.  

The proof of the last estimate in \eqref{Rem-est} is similar (cf. (B.8) and also (B.15) of \cite{HunSig1}). 
\end{proof}

Now, we prove various commutator estimates used in the main text. 
\begin{lemma}\label{lem:g-W-comm} \begin{align} \label{g-W-comm}[g(H), W_r]=O(r^{-\mu-1}).\end{align}\end{lemma} \begin{proof} As in \eqref{comm-exp1}, with $n=1$, we have
 \begin{align}
\label{comm-exp1'}[W_r, g(H)]&=\int d\widetilde g(z)(z-H)^{-1}[W_r, H](z-H)^{-1}. 
     \end{align} 
Using this, estimate \eqref{tildef-est} for $\widetilde g(z)$ and the fact that $[W_r, H]=\n W_r\cdot\n+\frac12
\Delta W_r$ is $S$-bounded, where the operator $S$ is given in \eqref{S-def}, and the estimate \[\|pR(z)\|\ls |\re z|^{1/2}/|\im z|\] (cf. \eqref{C1-est}), we arrive at \eqref{g-W-comm}. 
\end{proof}

\begin{lemma}\label{lem:chi-g+-est} 
Estimate \eqref{chi-g+-est}, that is, $\|\chi(x_{0s})g_{s^\beta}(H)\chi^-_b\| =O(s^{\beta-1})$, holds.
\end{lemma} 

\begin{proof}
Recall the definition $g_{s^\beta}(H):=U_{s^\beta}^{-1}g(H) U_{s^\beta}$ 
and let $\chi\equiv \chi(x_{0s})$. 
Using $\chi \chi^-_b=0$, we write 
\begin{align}\label{chi-g+-exp}
\chi g_{s^\beta}(H)\chi^-_b
 = [\chi, U_{s^\beta}^{-1}]g(H) U_{s^\beta}\chi^-_b+U_{s^\beta}^{-1}& [\chi, g(H)] U_{s^\beta}\chi^-_b\notag
 \\
& + U_{s^\beta}^{-1}g(H)[\chi,  U_{s^\beta}]\chi^-_b.
\end{align}
Since $[\chi,  U_{s^\beta}]= U_{s^\beta}(U_{s^\beta}^{-1} \chi U_{s^\beta}-\chi)
= U_{s^\beta}\int_0^{s^\beta}\p_r(U_r^{-1} \chi U_r)dr$ and $\p_r(U_r^{-1} \chi U_r)=i U_r^{-1} [H_r, \chi] U_r$ ,
we have
\begin{align}\label{chi-U-comm-exp'}
[\chi, U_{s^\beta}]=i U_{s^\beta}\int_0^{s^\beta} U_r^{-1} [H_r, \chi] U_rdr.
\end{align}
Note that $[H_r, \chi]=-i p\n\chi + \frac12(\Delta\chi)$. We control $p$ by $S^{-1}=(H+c)^{-1/2}$ (see \eqref{S-def}) as 
 \begin{align}\label{p-rep}p=S B=B' S,\end{align}  
 where $B:=(H+c)^{-1/2}p$ and $B':=p(H+c)^{-1/2}$, bounded operators.  Eq. \eqref{chi-U-comm-exp'}, together with the last two relations, gives
\begin{align}\label{chi-U-comm-exp''}
 [\chi, U_{s^\beta}]&= U_{s^\beta} \int_0^{s^\beta} U_r^{-1} \big(S B\n\chi
   + i\frac12(\Delta\chi) \big) U_r dr.
\end{align}
Next, we commute $(H+c)^{1/2}$ to the left. To this end, we apply the equation
\begin{align}\label{p-U-comm-est}
[S, U_{r}]=O(1),
\end{align}
which we now prove. 
First, we write $S=(H+c)^{1/2}=(H+c)(H+c)^{-1/2}$ and use the explicit formula 
$(H +c)^{-s}:=c'  \int_0^\infty (H +c+\om)^{-1} d\om/\om^s, $ 
 where $s\in (0, 1)$ and $c' :=[\int_0^\infty (1+\om)^{-1} d\om/\om^s]^{-1}$, to obtain $ [W_{r}, (H+c)^{1/2}]=O(r^{-\mu-1})$. This implies the estimate $ [H_{r}, (H+c)^{1/2}]= [W_{r}, (H+c)^{1/2}]=O(r^{-\mu-1})$, which, together with $\mu>0$  
and the relation 
\begin{align}\label{p-U-comm-est'}
[S, U_{r}]=U_{r}\int_0^{r}i U_{r'}^{-1} [H_{r'}, S] U_{r'}dr',
\end{align} 
  yields \eqref{p-U-comm-est}.

Commuting $S$ in Eq. \eqref{chi-U-comm-exp''} to the left (either twice through $U_r^{-1}$ and $U_{s^\beta}$, or once through $U_{s^\beta}U_r^{-1}=U(s^\beta, r)$) and using  \eqref{p-U-comm-est}, $\n\chi=O(s^{-1})$ and $\Delta\chi=O(s^{-2})$, gives
\begin{align}
\notag [\chi, U_{s^\beta}] \notag &=S O(s^{\beta-1})+  \int_0^{s^\beta}  \big( O(s^{-1})  
   + O(s^{-2}) \big) U_r dr\\   
\label{chi-U-comm-exp} 
& =S O(s^{\beta-1}) + O(s^{\beta-1}).
\end{align}

A similar estimate holds for $[\chi, U_{s^\beta}^{-1}]= - [\chi, U_{s^\beta}]^*$:
\begin{align}
\label{chi-Uinv-comm-est}
[\chi, U_{s^\beta}^{-1}]=O(s^{\beta-1})S+O(s^{\beta-1}).
\end{align}
 
Now, the second term on the r.h.s. of the above relation produces the right bound, $O(s^{-1+\beta})$ 
and so does the first term multiplied by $g(H)$, as $(H+1)^{1/2} g(H)$ is a bounded operator. 
This shows that the first term on the r.h.s. of \eqref{chi-g+-exp} is of the order $O(s^{-1+\beta})$. 
The same estimates apply to  the third term on the r.h.s. of \eqref{chi-g+-exp} giving $O(s^{-1+\beta})$. 
For the second term on the r.h.s. of \eqref{chi-g+-exp}, we use  \eqref{comm-exp-x} to obtain $[\chi, g(H)] =O(s^{-1})$.
This proves  \eqref{chi-g+-est}.  
\end{proof}

\end{document}